\documentclass[english]{article}
\usepackage[T1]{fontenc}
\usepackage[latin9]{inputenc}
\usepackage{geometry}
\usepackage{textcomp}
\usepackage{amsmath}
\usepackage{amsthm}
\usepackage{amssymb}
\usepackage{graphicx}
\PassOptionsToPackage{normalem}{ulem}
\usepackage{ulem}

\makeatletter

\providecommand{\tabularnewline}{\\}

\numberwithin{equation}{section}
\numberwithin{figure}{section}
\numberwithin{table}{section}
\theoremstyle{plain}
\newtheorem{thm}{\protect\theoremname}[section]
  \theoremstyle{definition}
  \newtheorem{defn}[thm]{\protect\definitionname}
  \theoremstyle{remark}
  \newtheorem{rem}[thm]{\protect\remarkname}
  \theoremstyle{plain}
  \newtheorem{cor}[thm]{\protect\corollaryname}
  \theoremstyle{plain}
  \newtheorem{lem}[thm]{\protect\lemmaname}

\makeatother

\usepackage{babel}
  \providecommand{\corollaryname}{Corollary}
  \providecommand{\definitionname}{Definition}
  \providecommand{\lemmaname}{Lemma}
  \providecommand{\remarkname}{Remark}
\providecommand{\theoremname}{Theorem}

\begin{document}

\title{Stable Support Recovery of Stream of Pulses with Application to Ultrasound
Imaging }

\date{August 2015, revised December 2015}

\maketitle

\author{\begin{center} Tamir Bendory, \footnotemark \and Avinoam Bar-Zion\footnotemark[\value{footnote}]\footnotetext{These two authors contributed equally to this work.}, Dan Adam, Shai Dekel and Arie Feuer\end{center}} 

\begin{abstract}
This paper considers the problem of estimating the delays of a weighted
superposition of pulses, called stream of pulses, in a noisy environment.
We show that the delays can be estimated using a tractable convex
optimization problem with a localization error proportional to the
square root of the noise level. Furthermore, all false detections
produced by the algorithm have small amplitudes. Numerical and in-vitro
ultrasound experiments corroborate the theoretical results and demonstrate
their applicability for the ultrasound imaging signal processing.
\end{abstract}

\section{Introduction }

In many engineering and scientific problems, we acquire data that
can be modeled as a weighted super-position of pulses (kernels), and
aim to decompose it into its building blocks, frequently called \emph{atoms}.
Typical examples are ultrasound imaging \cite{tur2011innovation,wagner2012compressed}
and radar \cite{Bar-IlanSub-Nyquis}, where the measurements are echoes
of the emitted pulse, that are reflected from different targets. When
attempting to decompose the stream of pulses there are two main concerns:
the robustness of the estimation and its degree of localization. While
the localization of the estimation determines how close is the estimation
to the real location of the targets, the robustness of the estimation
determines the degree to which reliable reconstruction can be performed
under different noise conditions. This paper investigates the localization
of the decomposition of stream of pulses with application to ultrasound
imaging.

Ultrasound imaging is currently one of the most widely-used medical
imaging modalities worldwide. The popularity of this modality can
be attributed to a variety of appealing characteristics, such as cost-effectiveness,
portability and being practically harmless. In the past few years
there is a shift in ultrasound systems from hardware based beam-forming
and signal processing to software based implementations. In these
new systems the RF data is accessible with greater ease, enabling
implementation of advanced signal processing algorithms. Although
ultrasound scans are sufficiently accurate for many applications,
the resolution of this modality is limited by the finite bandwidth
of ultrasound probes and the aperture of the transducer which determines
the axial and lateral resolution, respectively. The limited resolution
of ultrasound scans obscures important details and limits the clinical
useability of the modality. Another inherent feature of ultrasound
imaging, being a coherent imaging modality, is speckle noise. Biological
tissues are characterized by numerous structures much smaller than
the wavelength of the ultrasonic pulse. These structures induce local
inhomogeneity of the acoustic impedance. As a result, the ultrasonic
pulse is reflected from many independent scatterers in each resolution
cell, defined by the -6dB of the main beam in the axial, lateral and
out-of plane dimensions. The combination of these reflections creates
a complicated constructive and destructive interference pattern, named
speckle noise. Even though the small reflectors creating the speckle
noise cannot be resolved, the localization of strong reflectors is
important for many clinical application, motivating the development
of deconvolution and super-localization methods \cite{michailovich2005novel,adam2002blind,michailovich2007blind}. 

The ability of different biological structures to reflect the ultrasonic
wave is commonly represented by a 3D function, called the reflectivity
function. When assuming weak scattering and linear propagation, the
received signal (RF signal) can be modeled as a convolution of the
reflectivity function and the point spread function (PSF) of the ultrasound
scanner \cite{michailovich2003robust}. Even in cases where non-linear
propagation is dominant, the convolution model can be assumed to hold
for short signal sections. The automatic extraction of such sections
from the received RF data was discussed in several papers (e.g. \cite{michailovich2005novel,michailovich2003robust}).
Even though the ultrasonic signal results from a convolution of the
3D PSF and the 3D reflectivity function, the deconvolution problem
is usually solved by reducing the dimension of the problem and estimating
the hypothetical 1D reflectivity function from each scan line, called
a-line.

We consider a signal (scan line) which consists of a stream of pulses,
i.e. 
\begin{equation}
y\left[k\right]=\sum_{m}c_{m}g_{\sigma}\left[k-k_{m}\right]+\eta\left[k\right],\thinspace k\in\mathbb{Z},\thinspace c_{m}\in\mathbb{R},\label{eq:signal}
\end{equation}
where $g_{\sigma}\left[k\right]:=g\left(\frac{k}{\sigma N}\right)$,
$g(t)$ is an \emph{admissible} kernel, a notion which is defined
in Section \ref{sec:main_result}, $1/N$ is the uniform sampling
interval, $\sigma>0$ is a scaling parameter and $\eta$ is additive
noise with $\left\Vert \mathbf{\eta}\right\Vert _{_{1}}:=\sum_{k}\left|\eta_{k}\right|\leq\delta$
. We aim to estimate the true support $K:=\left\{ k_{m}\right\} $
and the weights $\left\{ c_{m}\right\} $ from the measurements $\left\{ y\left[k\right]\right\} $.
This paper does not deal with sampling techniques and any error due
to the sampling process (e.g. aliasing) can be treated as noise. 

The model (\ref{eq:signal}) can be represented as a convolution model
\begin{equation}
y\left[k\right]=\left(g_{\sigma}\ast x\right)\left[k\right]+\eta\left[k\right],\label{eq:alter_rep}
\end{equation}
where $'\ast'$ represents a discrete convolution, and 
\begin{equation}
x\left[k\right]=\sum_{m}c_{m}\delta\left[k-k_{m}\right],\label{eq:x}
\end{equation}
where $\delta\left[k\right]$ denotes a Kronecker delta function.
Our results can be applied directly to deconvolve richer classes of
signals, as elaborated in Section \ref{sec:main_result}.

A well-known approach to decompose the signal into its atoms is by
using parametric methods such as Prony, MUSIC, Matrix Pencil and ESPRIT
\cite{stoica2005spectral,schmidt1986multiple,Matrix_pencil,roy1989esprit}.
However, the stability of these methods is not well-understood, although
some steps towards the study of the non-asymptotic behavior have been
taken recently \cite{liao2014music,moitra2014threshold}. An alternative
way is to exploit standard compressed sensing and sparse representations
theorems, relying on the sparsity of the signal \cite{elad2010sparse,donoho2006compressed}.
These methods require dictionaries with low coherence, and hence highly
limiting the resolution of the solution.\textbf{ }

Recently, it was shown that under a kernel-dependent separation condition
on the support $K$ (see Definition \ref{def:separation}) a signal
of the form (\ref{eq:x}) can be estimated robustly by solving an
$\ell_{1}$ minimization problem \cite{bendorySOP}. This fundamental
result is presented in Theorem \ref{th:noisy}. This paper builds
on the results of \cite{bendorySOP} and follows \cite{support_detection}
to derive new results, revealing that the recovery of the $\ell_{1}$
minimization is well-localized. That is to say, the support of the
recovered signal is clustered around the support of the underlying
signal (\ref{eq:x}) as presented in Theorem \ref{th:main}. The proof
relies on the existence of a special interpolating function, constructed
as a super-position of the kernel and its derivatives. Similar techniques
were developed in the field of sparse recovery, and specifically as
the pillars of the super-resolution theory \cite{candes2013super,candes2013towards,tang2013compressed,de2012exact,bendory2013exact,bendory2013Legendre,bendorySHalgorithm,bendoryPositiveSOP,7286741}. 

As part of the increasing interest in the field, several previous
works explored the localization properties of super-resolution problems.
In \cite{de2014non,azais2014spike,support_detection}, the authors
focused their attention merely on the special cases of low-passed
convolution kernels in Fourier and algebraic polynomials spaces. In
this work we treat a much broader family of convolution kernels, named
admissible kernels (see Definition \ref{def:kernel}). The thorough
works \cite{duval2013exact,duval2015sparse} investigated the localization
behavior when the noise level drops to zero. In contrast, the analysis
in this paper, as presented in Theorem \ref{th:main}, holds for all
noise levels. Before proceeding, we would like to mention that support
recovery of sparse signals was studied in many different settings
and approaches. For instance, \cite{wainwright2009information,akccakaya2010shannon}
analyzed the support recovery of compressed sensing problems using
an information-theory point of view.

We validate our theoretical results using ultrasound imaging experiments.
These experiments include both numerical simulations and in-vitro
phantom scans. The numerical simulations enable the validation of
the main result presented in this paper using numerous reflectivity
functions with varying spacings and amplitudes. On the other hand,
the phantom scans were used to test this method under the conditions
of real ultrasound signals and a reflectivity function with known
geometry. Together these experiments corroborate the theoretical results
and show their relevance to ultrasound signal deconvolution.

The outline of this paper is as follows. In section \ref{sec:main_result},
we recall the results of \cite{bendorySOP} and present our main theoretical
result. Section \ref{sec:proof} is devoted to the proof of the localization
result. Section \ref{sec:exp} describes the setup of the experiments
we have conducted and presents the results, and Section \ref{sec:Conclusions}
concludes the work and outlines some possible future research directions.

\section{Main Result\label{sec:main_result}}

In a previous work \cite{bendorySOP}, it was shown that robust recovery
of a signal (\ref{eq:x}) from a stream of pulses (\ref{eq:alter_rep})
is possible if the support $K$ is sufficiently separated and the
convolution kernel $g(t)$ meets some mild localization properties.
To this end, we recall two interrelated fundamental definitions. A
signal that its support does not satisfy Definition \ref{def:separation}
cannot be decompose stably . For instance, a signal of the form of
$y\left[k\right]=g_{\sigma}\left[k-k_{1}\right]-g_{\sigma}\left[k-k_{2}\right]$
where $k_{1}$ and $k_{2}$ are adjacent locations on the grid will
be drown in a miniscule noise level unless the kernel $g_{\sigma}$
is very localized. Definition \ref{def:kernel} formulates the localization
properties of the kernel. Roughly speaking, it requires that the kernel
and its first three derivatives will decay sufficiently fast and that
it is concave near the origin.
\begin{defn}
\label{def:separation}A set of points $K\subset\mathbb{Z}$ is said
to satisfy the minimal separation condition for a kernel dependent
$\nu>0$ and given scaling parameter $\sigma>0$ and sampling spacing
$1/N>0$ if 
\[
\min_{k_{i},k_{j}\in K,i\neq j}\left|k_{i}-k_{j}\right|\geq N\nu\sigma.
\]

\end{defn}

\begin{defn}
\label{def:kernel}A kernel $g$ is admissible if it has the following
properties:\end{defn}
\begin{enumerate}
\item $g\in\mathcal{C}^{3}(\mathbb{R})$, is real and even.
\item \uline{Global property:} There exist constants $C_{\ell}>0,\ell=0,1,2,3$
such that $\left|g^{(\ell)}(t)\right|\leq C_{\ell}/\left(1+t^{2}\right)$,
where $g^{(\ell)}(t)$ denotes the $\ell^{th}$derivative of $g$.
\item \uline{Local property:} There exist constants $0<\varepsilon,\beta<\nu$
such that

\begin{enumerate}
\item $g(t)>0$ for all $\left|t\right|\leq\varepsilon$ and $g(t)<g(\varepsilon)$
for all $\left|t\right|>\varepsilon$. 
\item $g^{(2)}(t)<-\beta$ for all $\left|t\right|\leq\varepsilon$.
\end{enumerate}
\end{enumerate}
\begin{rem}
Two prime examples for admissible kernels are the Gaussian kernel
$g(t)=e^{-\frac{t^{2}}{2}}$ and Cauchy kernel $g(t)=\frac{1}{1+t^{2}}$.
The numerical constants associated with those kernel are presented
in Table \ref{tab1}. 
\end{rem}

\begin{rem}
The global property in Definition \ref{def:kernel} can be somewhat
weakened to $\left\vert g^{\left(\ell\right)}\left(t\right)\right\vert \leq C_{\ell}/(1+\vert t\vert^{1+s})$
for some $s>0$. In this case, the separation condition would become
dependent on $s$.

\begin{table}
\begin{centering}
\begin{tabular}{|c|c|c|}
\hline 
Kernels & Gaussian$:=e^{-\frac{t^{2}}{2}}$ & Cauchy$:=\frac{1}{1+t^{2}}$\tabularnewline
\hline 
\hline 
$C_{0}$ & 1.22 & 1\tabularnewline
\hline 
$C_{1}$ & 1.59 & 1\tabularnewline
\hline 
$C_{2}$ & 2.04 & 2\tabularnewline
\hline 
$C_{3}$ & 2.6 & 5.22\tabularnewline
\hline 
$g^{(2)}(0)$ & -1 & -2\tabularnewline
\hline 
empirical $\nu$ & 1.1 & 0.45\tabularnewline
\hline 
\end{tabular}
\par\end{centering}

\caption{\label{tab1}\emph{The table presents the numerical constants of the
global property in Definition \ref{def:kernel} for the Gaussian and
Cauchy kernels. Additionally, we evaluated by numerical experiments
the minimal empirical value of $\nu$, the separation constant of
Definition \ref{def:separation} for each kernel.}}
\end{table}

\end{rem}
The kernel-dependent separation constant $\nu$ is defined as the
minimal constant guaranteeing the existence of some interpolating
function as described in Section 3.1 of \cite{bendorySOP} (see Lemma
\ref{lemma:q}). The minimal separation constant for the Cauchy and
Gaussian kernels was evaluated numerically to be $0.45$ and $1.1$
(see Table \ref{tab1}). 

In this paper, we consider recovery by solving the following tractable
convex optimization problem:
\begin{equation}
\min_{\tilde{x}\in\ell_{1}(\mathbb{Z})}\left\Vert \tilde{x}\right\Vert _{_{1}}\quad\mbox{subject to}\quad\left\Vert y-\left(g_{\sigma}\ast\tilde{x}\right)\right\Vert _{_{1}}\leq\delta.\label{eq:opt_noisy}
\end{equation}
 Denote the solution of (\ref{eq:opt_noisy}) as 
\begin{equation}
\hat{x}[k]=\sum_{n}\hat{c}_{n}\delta\left[k-\hat{k}_{n}\right],\quad\hat{K}:=\left\{ \hat{k}_{n}\right\} .\label{eq:x_hat}
\end{equation}

We are ready now to state the main result of \cite{bendorySOP}. This
result shows that the recovery error of (\ref{eq:opt_noisy}) is proportional
to noise level. 
\begin{thm}
\label{th:noisy}Let $y$ be of the form of (\ref{eq:signal}) and
let $g$ be an admissible kernel as defined in Definition \ref{def:kernel}.
If the signal's support $K$ satisfies the separation condition of
Definition \ref{def:separation} for $\sigma,N>0$, then the solution
$\hat{x}$ of (\ref{eq:opt_noisy}) satisfies (for sufficiently large
$\nu)$ 
\[
\left\Vert \hat{x}-x\right\Vert _{_{1}}\leq\frac{16\gamma^{2}}{\beta}\delta,
\]
where 
\begin{eqnarray}
\gamma & = & \max\left\{ N\sigma,\varepsilon^{-1}\right\} .\label{eq:gamma}
\end{eqnarray}
\end{thm}
\begin{cor}
In the noiseless case $\delta=0$, (\ref{eq:opt_noisy}) results in
exact recovery.
\end{cor}

\begin{rem}
A tighter bound on the recovery error is presented in Theorem 2.12
in \cite{bendorySOP}. For the sake of clarity, we presented here
a simplified version of the stability result. 
\end{rem}

\begin{rem}
\label{rem:2d}We chose to focus on the univariate results as we aim
to apply it on 1D ultrasound scans. However, we stress that an equivalent
result holds for the bivariate case. Namely, the recovery error is
proportional to the noise level $\delta$ and depends on the kernel
localization properties. 
\end{rem}
For a fix $k_{m}\in K$, let 
\begin{equation}
\mathbb{Z}_{m}:=\left\{ k\in\mathbb{Z}:\thinspace\left|k-k_{m}\right|\leq N\varepsilon\sigma\right\} .\label{eq:Zm}
\end{equation}
Accordingly, we define the partition of $\mathbb{Z}$ to disjoint
subsets $\mathbb{Z=Z}_{near}\cup\mathbb{Z}_{far}$, where 
\begin{eqnarray}
\mathbb{Z}_{near} & := & \cup_{\left\{ m:\thinspace k_{m}\in K\right\} }\mathbb{Z}_{m},\label{eq:sets}\\
\mathbb{Z}_{far} & := & \mathbb{Z}_{near}^{C}\nonumber \\
 & = & \left\{ k\in\mathbb{Z}:\thinspace\left|k-k_{m}\right|>N\varepsilon\sigma,\forall k_{m}\in K\right\} .\nonumber 
\end{eqnarray}
As the notation implies and since $\varepsilon<\nu$, the set $\mathbb{Z}_{near}$
is composed of the grid points that are close to a support location,
whereas $\mathbb{Z}_{far}$ is composed of all grid points that are
located sufficiently far from all support locations. 

This paper focuses on the localization properties of the solution
of (\ref{eq:opt_noisy}). Therefore, we present the following important
corollary which states that the superfluous spikes, far away from
the sought signal, have low amplitudes. 
\begin{cor}
\label{cor:artifacts}Let us define the set $\hat{K}_{far}:=\mathbb{Z}_{far}\cap\hat{K}$,
where $\mathbb{Z}_{far}$ and \textup{$\hat{K}$} are given in (\ref{eq:sets})
and (\ref{eq:x_hat}), respectively. Under the conditions of Theorem
\ref{th:noisy}, we have 
\[
\sum_{\left\{ n:\thinspace\hat{k}_{n}\in\hat{K}_{far}\right\} }\left|\hat{c}_{n}\right|\leq\frac{16\gamma^{2}}{\beta}\delta.
\]
\end{cor}
\begin{proof}
By definition, $x[k]=0$ for all $k\in\hat{K}_{far}$. Therefore,
using Theorem \ref{th:noisy} we get 
\begin{eqnarray*}
\sum_{\left\{ n:\thinspace\hat{k}_{n}\in\hat{K}_{far}\right\} }\left|\hat{c}_{n}\right| & = & \sum_{\left\{ n:\thinspace\hat{k}_{n}\in\hat{K}_{far}\right\} }\left|\hat{x}\left[\hat{k}_{n}\right]-x\left[\hat{k}_{n}\right]\right|\\
 & \leq & \sum_{k\in\mathbb{Z}}\left|\hat{x}\left[k\right]-x\left[k\right]\right|\\
 & \leq & \frac{16\gamma^{2}}{\beta}\delta.
\end{eqnarray*}

\end{proof}
Although the optimization problem (\ref{eq:opt_noisy}) is defined
on a discrete grid, it was proven that it converges to a solution
on the continuum (in the sense of measures) as the discretization
becomes finer \cite{just_dis}. Additionally, the behavior of the
discrete optimization problem solution, when the underlying signal
is defined on a compact domain is analyzed in detail in \cite{duval2013exact,duval2015sparse}.

We want to emphasize that our model can be extended to other types
of underlying signals, not necessarily a spike train as in (\ref{eq:x}).
For instance, suppose that the underlying signal itself is a stream
of pulses, namely a signal of the form $\grave{x}[k]=\sum c_{m}\grave{g}_{\sigma_{2}}\left[k-k_{m}\right]$.
In this case, the measurements are given as $y\left[k\right]=\left(g_{\sigma_{1}}\ast\grave{g}_{\sigma_{2}}\ast x\right)[k]$
where $x$ is a signal of the form of (\ref{eq:x}) and $g_{\sigma_{1}}$
is the convolution kernel. Therefore, our results hold immediately
if the combined convolution kernel $\tilde{g}\left[k\right]=\left(g_{\sigma_{1}}\ast\grave{g}_{\sigma_{2}}\right)\left[k\right]$
meets the definition of admissible kernel and the signal's support
$K$ satisfies the associated separation condition. For instance,
if $g$ and $\grave{g}$ are both Gaussian kernels with standard deviations
of $\sigma_{1}$ and $\sigma_{2}$, then $\tilde{g}$ is also Gaussian
with standard deviation of $\sigma=\sqrt{\sigma_{1}^{2}+\sigma_{2}^{2}}$
and thus obeys the definition of admissible kernel. 

The main theoretical contribution of this paper is the following theorem,
stating that under the separation condition, for any $k_{m}\in K$
with sufficiently large amplitude $c_{m}$, there exists a close $\hat{k}_{m}\in\hat{K}$.
Namely, the solution of (\ref{eq:opt_noisy}) locates a spike near
any spike of the underlying signal. In \cite{support_detection},
the support detection for the super-resolution problem was examined
by convex optimization tools. We follow this line of research and
present the main theorem of this paper, as follows:
\begin{thm}
\label{th:main}Let $y$ be of the form of (\ref{eq:signal}) and
let $g$ be an admissible kernel as defined in Definition \ref{def:kernel}.
If the signal support $K$ satisfies the separation condition of Definition
\ref{def:separation}, then the solution $\hat{x}$ of the convex
program (\ref{eq:opt_noisy}) has the following property: For any
$m$ such that $\left|c_{m}\right|>\frac{16\gamma^{2}}{\beta}\delta$
there exists \textup{a $\hat{k}_{m}\in\hat{K}$} such that 
\begin{eqnarray*}
\left|k_{m}-\hat{k}_{m}\right| & \leq & \frac{8\gamma^{2}}{\beta}\sqrt{\frac{g(0)\delta}{\left(\left|c_{m}\right|-\frac{16\gamma^{2}}{\beta}\delta\right)}}.
\end{eqnarray*}

\end{thm}
The localization result of Theorem \ref{th:main} cannot be derived
directly from the robustness result of Theorem \ref{th:noisy}. For
instance, suppose that the amplitude $c_{m}$ and the noise level
$\delta$ are of the same order. Then, Theorem \ref{th:main} guarantees
that there exists at least one location $\hat{k}_{m}\in\hat{K}$ obeying
$\left|\hat{k}_{m}-k_{m}\right|\leq C$ for some kernel-dependent
constant $C$. This localization error cannot be inferred from the
robustness result.

\section{Proof of Theorem \ref{th:main} \label{sec:proof}}

The main pillar of the proof is the existence of a special superposition
of the kernel and its first derivative which satisfies some interpolation
properties. This function is frequently called the \emph{dual certificate}.
Next, we need the following Lemma from \cite{bendorySOP}:
\begin{lem}
\label{lemma:q} Suppose that the set $T:=\left\{ t_{m}\right\} \subset\mathbb{R}$
obeys the separation condition 
\[
\min_{t_{i},t_{j}\in T,i\neq j}\left|t_{i}-t_{j}\right|\geq\nu\sigma,
\]
for some constant $\nu>0$. Let $\left\{ u_{m}\right\} \in\pm1$ be
a sign sequence and let $g(t)$ be an admissible kernel as defined
in Definition \ref{def:kernel}. Then, there exists a set of coefficients
$\left\{ a_{m}\right\} $ and $\left\{ b_{m}\right\} $ such that
the function 
\[
q(t)=\sum_{m}a_{m}g\left(\frac{t-t_{m}}{\sigma}\right)+b_{m}g^{(1)}\left(\frac{t-t_{m}}{\sigma}\right),
\]
satisfying $\left\Vert q\right\Vert _{\infty}=\max_{t\in\mathbb{R}}\left|q(t)\right|\leq1$
and interpolates the sequence $\left\{ u_{m}\right\} $ on $\left\{ t_{m}\right\} $,
namely $q\left(t_{m}\right)=u_{m}$ for all $t_{m}\in T$. Additionally,
for all $t$ satisfying $\left|t-t_{m}\right|\leq\varepsilon\sigma$
for some \textup{$t_{m}\in T$, we have $\left|q(t)\right|\leq1-\frac{\beta\left(t-t_{m}\right)^{2}}{4g(0)\sigma^{2}}$.} 
\end{lem}
We will also need the following result:
\begin{lem}
\label{lemma:2.1} Let the signal support K satisfy the separation
condition of Definition \ref{def:separation} and let $\hat{K}_{near}:=\mathbb{Z}_{near}\cap\hat{K}$,
where $\mathbb{Z}_{near}$ and $\hat{K}$ are defined in (\ref{eq:sets})
and (\ref{eq:x_hat}), respectively. Then, 
\[
\sum_{\left\{ n:\thinspace\hat{k}_{n}\in\hat{K}_{near}\right\} }\left|\hat{c}_{n}\right|d^{2}\left(\hat{k}_{n},K\right)\leq\frac{64g(0)\gamma^{4}}{\beta^{2}}\delta,
\]
where $d\left(k,K\right):=\min_{k_{m}\in K}\left|k_{m}-k\right|$
and $\gamma$ and $\delta$ are defined in (\ref{eq:gamma}) and (\ref{eq:signal}),
respectively. \end{lem}
\begin{proof}
See Appendix \ref{sec:appA}.
\end{proof}
We are now ready to prove Theorem \ref{th:main}. Fix $k_{m}\in K$
, let $\mathbb{Z}_{m}$ be as in (\ref{eq:Zm}) and let $\hat{K}_{m}:=\mathbb{Z}_{m}\cap\hat{K}.$
Since $\varepsilon<\nu$ (the separation constant), $\mathbb{Z}_{m}\cap K=\left\{ k_{m}\right\} $.
Hence by Theorem \ref{th:noisy} we get 
\begin{eqnarray*}
\left|c_{m}-\sum_{\left\{ n:\thinspace\hat{k}_{n}\in\hat{K}_{m}\right\} }\hat{c}_{n}\right| & = & \left|\sum_{k\in\mathbb{Z}_{m}}x[k]-\hat{x}[k]\right|\\
 & \leq & \sum_{k\in\mathbb{Z}_{m}}\left|x[k]-\hat{x}[k]\right|\\
 & \leq & \sum_{k\in\mathbb{Z}}\left|x[k]-\hat{x}[k]\right|\\
 & \leq & \frac{16\gamma^{2}}{\beta}\delta.
\end{eqnarray*}
By the triangle inequality we then get $\left|c_{m}\right|-\frac{16\gamma^{2}}{\beta}\delta\leq\sum_{\left\{ n:\thinspace\hat{k}_{n}\in\hat{K}_{m}\right\} }\left|\hat{c}_{n}\right|$.
If $\left|c_{m}\right|>\frac{16\gamma^{2}}{\beta}\delta$, then $\hat{K}_{m}$
is not an empty set, namely, there exists at least one $\hat{k}_{n}\in\hat{K}$
so that $\left|\hat{k}_{n}-k_{m}\right|\leq N\varepsilon\sigma$.
Let $\bar{d}=\min_{\hat{k}_{n}\in\hat{K}_{m}}\left|\hat{k}_{n}-k_{m}\right|$.
Then, using Lemma \ref{lemma:2.1} we directly get 
\begin{eqnarray*}
\bar{d}^{2} & \leq & \frac{64g(0)\gamma^{4}}{\beta^{2}\sum_{\left\{ n:\thinspace\hat{k}_{n}\in\hat{K}_{near}\right\} }\left|\hat{c}_{n}\right|}\delta\\
 & \leq & \frac{64g(0)\gamma^{4}}{\beta^{2}\left(\left|c_{m}\right|-\frac{16\gamma^{2}}{\beta}\delta\right)}\delta.
\end{eqnarray*}
This concludes the proof.

\section{Ultrasound Experiments \label{sec:exp}}

The theoretical results and aforementioned algorithms were validated
in a series of numerical and \emph{in-vitro} experiments. In the \emph{in-vitro}
experiments, a mechanical phantom with a known geometry was scanned
using a single-element ultrasound transducer. The same setup was used
in order to estimate the PSF for the numerical experiments. An unfocused,
single-element Panametrix transducer with 0.375'' diameter was used
in all the scans (A3265-SU, Whaltham, MA). This transducer was operated
in a transmitter-receiver pulse-echo configuration. A single channel
arbitrary function generator (8024, Tabor Electronics, Tel-Hanan,
Israel) was used to drive the transducer with a pulse containing adjustable
number of sinusoidal cycles. The function generator was operated at
4MHz during all the experiments described below. The received RF-lines
were sampled at a rate of 100MHz using a 14bit digitizer (CS8427 ,Gage,
Lockport, IL). In order to protect the digitizer from over-voltage
the received RF data was passed through an external diplexer (RDX-6,
RITEC, Warwick, RI). A diagram of the experimental setup can be found
in Figure~\ref{fig:Exprimental-set-up:}a .The recorded RF signals
were loaded into a personal computer and analyzed using dedicated
MATLAB\textregistered{} code (The Mathworks, Natic, MA).

\subsection{Numerical Simulation\label{sub:Numerical-Simulation}}

In order to emulate a realistic RF line, the ultrasound pulse used
in the numerical simulations was the pulse-echo signal reflected from
a thick perspex sheet immersed in a water tank. Pulses with two and
three cycles were transmitted and the reflected echos were received
and recorded. It is assumed that the received pulse measured this
way is a reliable estimation for the convolution kernel in equation
(\ref{eq:alter_rep}). The simulated ultrasound signal was produced
according to the following steps. First, the simulated reflectivity
signal was produced. The reflectivity signals were comprised of a
series of evenly-spaced reflectors with random amplitudes uniformly
distributed between $\pm[5,10]$. Different separations enabled the
comparison between the localization results achieved for different
reflectors density. Second, the simulated reflectivity function was
convolved with the estimated kernel. Different levels of SNR were
simulated by adding a Gaussian white noise to the resulting signal.
100 simulated signals were produced for each combination of separation
and SNR conditions. 

Although the RF pulses as modulated signals do not decay monotonously,
the demodulated signals and their envelops are compatible with the
model described above (\ref{eq:signal}). Therefore, the simulated
signals were demodulated by multiplication with the carrier frequency
and application of a low-pass FIR filter. Following the demodulation,
the optimization problem (\ref{eq:opt_noisy}) was solved using CVX,
a toolbox for specifying and solving convex programs \cite{cvx}.
The value of $\delta$ was selected so that $\left\Vert \eta\ast g\right\Vert _{\ell_{1}}\leq\delta$
in accordance with equation (\ref{eq:opt_noisy}) and adjusted according
to the SNR level. Examples of the original and reconstructed reflectivity
functions produced using a two-cycle pulse and a three-cycle pulse
with SNR of 23 dB are presented in Figures \ref{fig:Simulated-signal}(a)
and \ref{fig:Simulated-signal}(b), respectively. This example demonstrates
how the localization of the estimation decreases and the locations
of the recovered reflectors get further from their original locations
as the width of the convolution kernel increases.

\begin{figure*}
\includegraphics[scale=0.7]{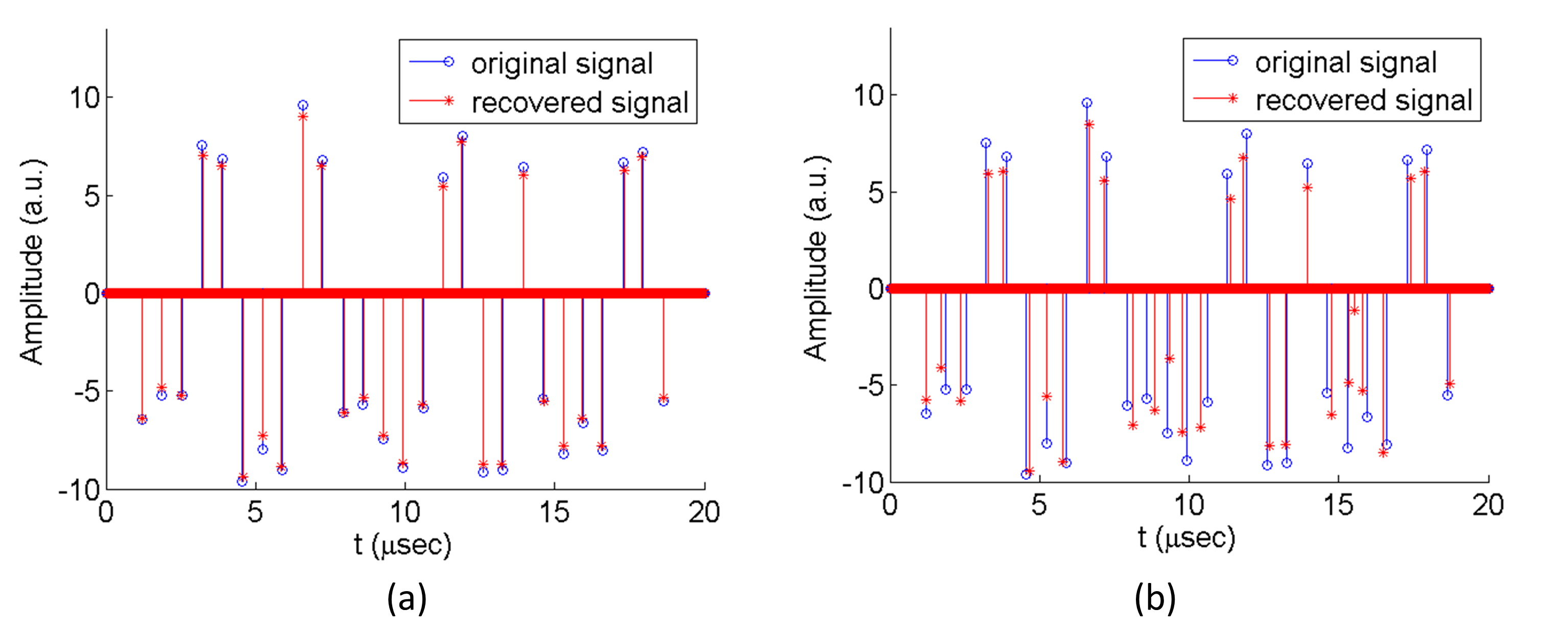}

\caption{\label{fig:Simulated-signal}Representative results of numerical simulation
experiments: reflectivity estimates are presented against true reflectors
locations, for two different pulse lengths at SNR = 23dB. (a) Results
from a two-cycle pulse, (b) results from a three-cycle pulse. }
\end{figure*}

In order to quantify the performance of the localization algorithm,
the localization error was measured for different levels of SNR using
two-cycle three-cycle pulses. The localization error was defined as
the distance between the real location of each reflector and the location
of the closest reflector in the solution. The results of this analysis
can be seen in Figure \ref{fig:Quantification-of-numerical}. For
each separation and SNR level, subplots (a) and (b) present the mean
localization error of the simulations using two-cycle and three-cycle
pulses, respectively. The mean localization error of the simulations
produced using the two-cycle pulse is relatively small for SNR levels
greater than 15 dB and increases for small separation intervals. In
contrast, the mean localization error of the simulations produced
using the three-cycle pulse is markedly higher as the reflectors are
not sufficiently separated. Subplots (c) and (d) present the standard
deviation of the localization error for each separation and SNR level.\textbf{
}When measuring the mean localization error most of the localization
error is averaged out, indicating that the solutions are clustered
around the real locations. The measured standard deviation values
are almost one order of magnitude higher than these of the mean localization
error. Hence, the standard deviation of the localization error is
a much more suitable criterion for localization loss. 

\begin{figure*}
\begin{centering}
\includegraphics[scale=0.7]{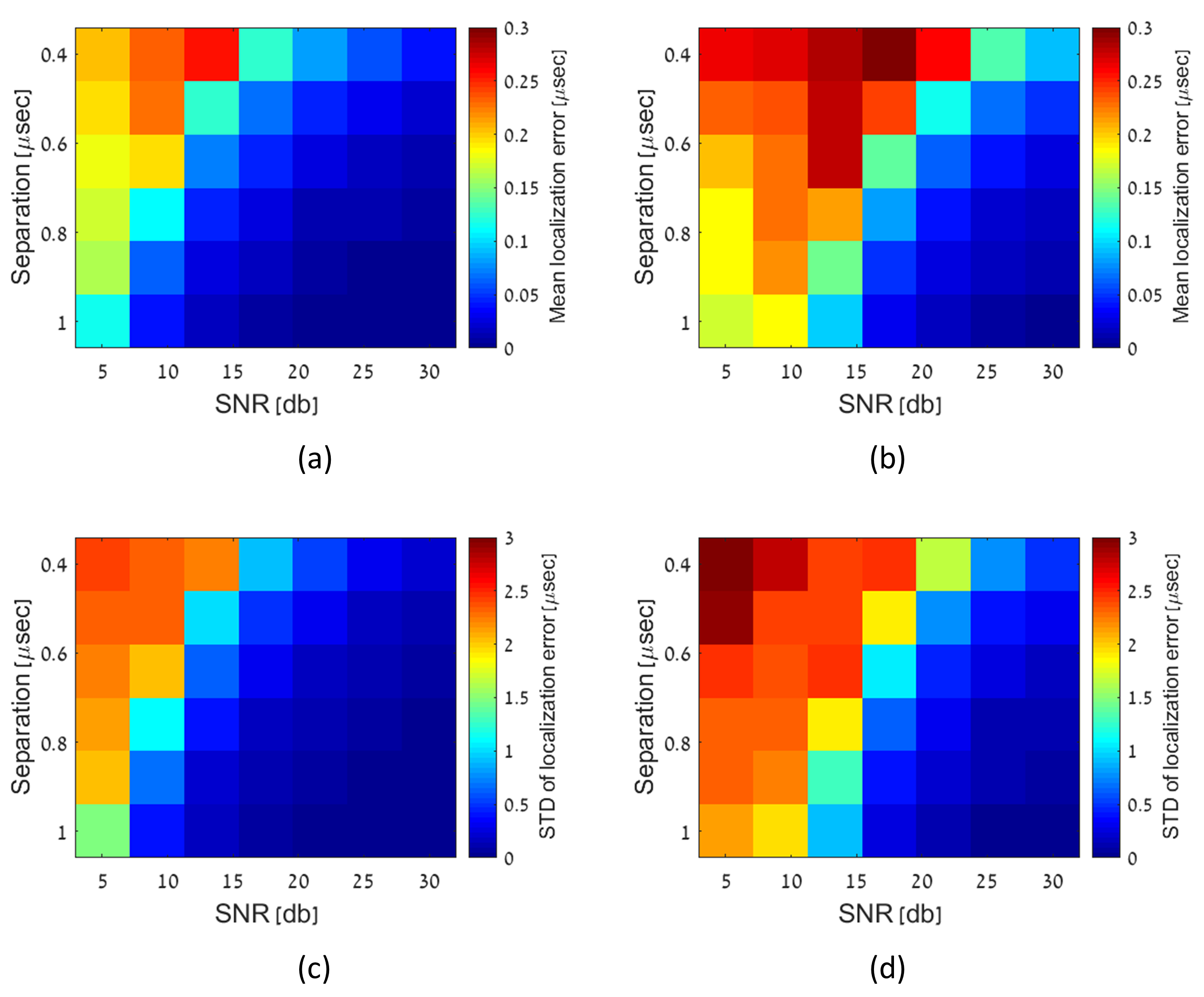}
\par\end{centering}

\caption{\label{fig:Quantification-of-numerical}Quantification of the numerical
simulations: In subplots (a) and (b) the mean localization error is
presented for each separation and SNR level for simulations produced
using a two-cycle pulse and a three-cycle pulse, respectively. In
subplots (c) and (d) the standard deviation of the localization error
is presented for simulations produced using a two-cycle pulse and
a three-cycle pulse, respectively. }
\end{figure*}

\subsection{In-vitro experiment}

In order to validate the algorithm's ability to reliably detect the
location of reflectors in real ultrasound RF scans, a phantom with
a known geometry was constructed. Similar to the phantom used in \cite{adam2002blind},
the phantom was constructed from parallel layers of perspex and polycarbonate
held together by a rigid frame. The dimensions of the phantom in the
direction orthogonal to the transmitted beam are much larger than
each layer's thickness making it a 1D phantom. The thickness of the
layers and the distance between them were measured in advance and
are detailed in Table \ref{tab2}. The large differences between the
impedance of the solid perspex and polycarbonate sheets and the water
in which they are submerged produce strong acoustic reflectors in
each interface. Therefore, only a fraction of the acoustic energy
penetrates beyond the first few layers. Due to the fast decay of the
signal only the first five layers (six reflectors) were relevant to
our analysis. 

\begin{table}
\begin{centering}
\begin{tabular}{|c|c|c|}
\hline 
Number of layer & Layer Thickness (mm) & Layer Spacing (mm) \tabularnewline
\hline 
\hline 
1 & 3 (Perspex) & \tabularnewline
\hline 
 &  & 1.3\tabularnewline
\hline 
2 & 1 (polycarbonate) & \tabularnewline
\hline 
 &  & 0.9\tabularnewline
\hline 
3 & 1.8 (Perspex) & \tabularnewline
\hline 
 &  & 1.3\tabularnewline
\hline 
4 & 1 (polycarbonate) & \tabularnewline
\hline 
 &  & 1.5\tabularnewline
\hline 
5 & 1.8 (Perspex) & \tabularnewline
\hline 
 &  & 1.3\tabularnewline
\hline 
6 & 3 (Perspex) & \tabularnewline
\hline 
\end{tabular}
\par\end{centering}

\caption{\label{tab2}Dimensions of the Perspex-polycarbonate phantom. The
gap between each Perspex/ polycarbonate layer was filled with water
as the phantom was submerged in a water tank.}
\end{table}

The phantom was design to enable the validation of the localization
algorithm under a various of relevant conditions. The thickness of
the Perspex and polycarbonate layers and the distance between them
are of the same order as the wavelength of the transmitted pulse.
Therefore, the RF-lines can have variable levels of echo interference
as function of the number of cycles transmitted in each pulse. Additionally,
although the alternating high reflectivity values do not mimic a specific
tissue of organ, they enable the validation of the localization algorithm
under various SNR values. 

The transducer was inserted through one of the walls of the a dedicated
water tank filled with deionized water at room temperature (see Figure
\ref{fig:Exprimental-set-up:}(b)). The phantom was placed in the
center of the tank and was aligned to be perpendicular to the beam
by changing its orientation until a maximal echo amplitude was achieved.
Pulses with one to four cycles were transmitted and the reflected
echos were received and recorded. The pulses containing different
number of transmitted cycles produced signals with slightly different
peak amplitudes. In order to achieve comparable base-ground SNR level,
white Gaussian noise was added to the received signal resulting in
identical SNR of 41dB. The pulse shape was estimated from the echo
reflected from the first Perspex layer, which was thick enough to
ensure the separation of the first echo from the consecutive reflections.
Later, different levels of SNR were achieved by adding white Gaussian
noise with different standard deviation to the resulting signals. 

\begin{figure*}
\begin{centering}
\includegraphics[scale=0.7]{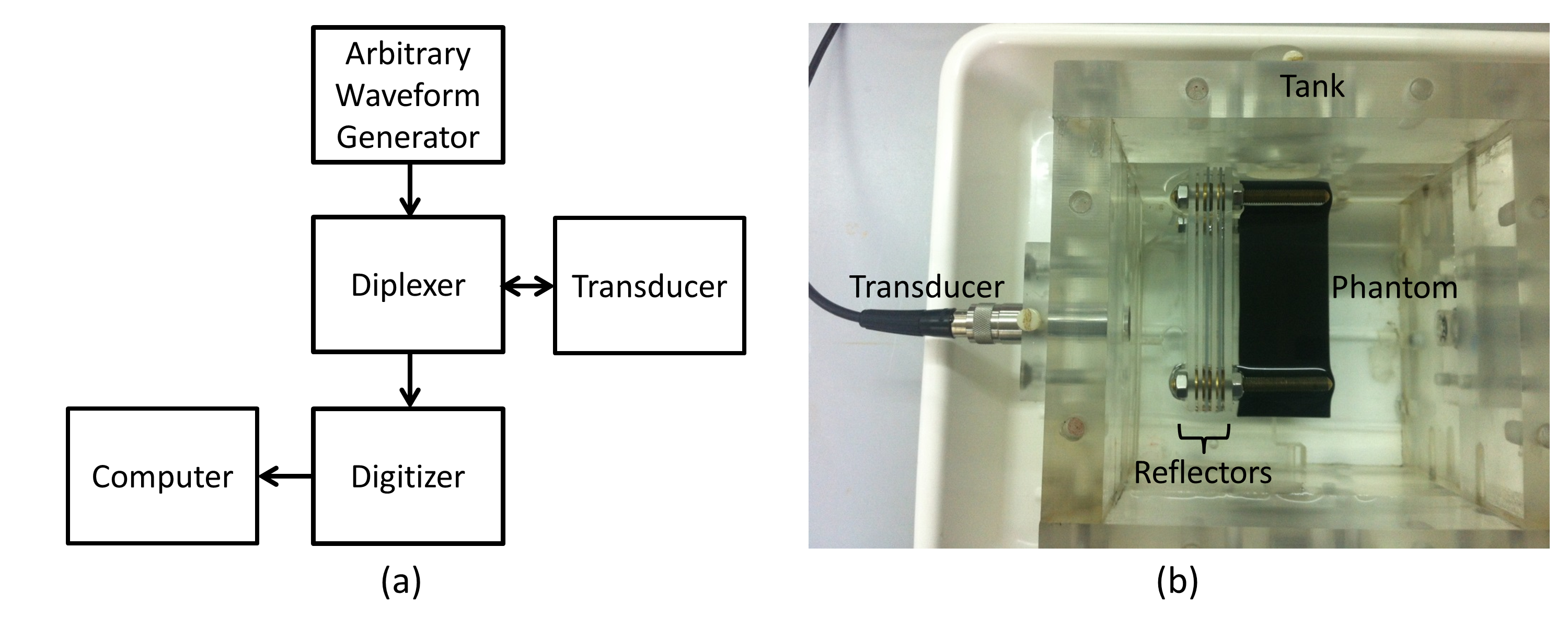}
\par\end{centering}

\caption{\label{fig:Exprimental-set-up:}Experimental set-up. (a) a diagram
of the experimental setup. (b) the mechanical phantom placed in front
of the transducer. }
\end{figure*}

The RF signals are demodulated as described in Section \ref{sub:Numerical-Simulation}.
After the demodulation, the optimization problem (\ref{eq:opt_noisy})
was solved using CVX \cite{cvx}. Complex reflectivity function values
were enabled in order to facilitate the incorporation of phase delays
into the signal model as proposed in \cite{wagner2012compressed}.
The value of the noise level parameter $\delta$ was selected so that
$\left\Vert \eta\ast g\right\Vert _{\ell_{1}}\leq\delta$ in accordance
with equation (\ref{eq:opt_noisy}) and adjusted according to the
SNR level. The reflectivity function estimates were compared against
the true location of the reflectors used as a reference. The absolute
value of the estimated reflectivity function is presented along with
the location of the reflectors. Since both the attenuation and acoustic
reflections gradually reduce the energy of the traveling wave, there
was no attempt to compare the estimated magnitude of the reflectivity
function to that of the phantom.

Two RF-lines received after transmitting a single-cycle pulse and
a four-cycle pulse are presented in Figures \ref{fig:Results-of-phantom}(a)
and \ref{fig:Results-of-phantom}(b), respectively. The estimated
reflectivity function correlates well with the structure of the phantom
and the different interfaces originally masked by the interference
of the pulses are revealed. When wider pulses are transmitted, the
kernel-dependent minimal separation condition demands greater separation
of the strong reflectors. As a result, better localization of the
reflectivity estimates are achieved for shorter pulses, containing
one or two cycles. Since the reflected energy is not compensated after
each layer transition, the amplitude of the estimated reflectivity
function decrease with depth. As the pulse energy decrease with depth,
the local SNR decreases, resulting in reduced localization of deeper
phantom layers.

\begin{figure*}
\centering{}\includegraphics[scale=0.7]{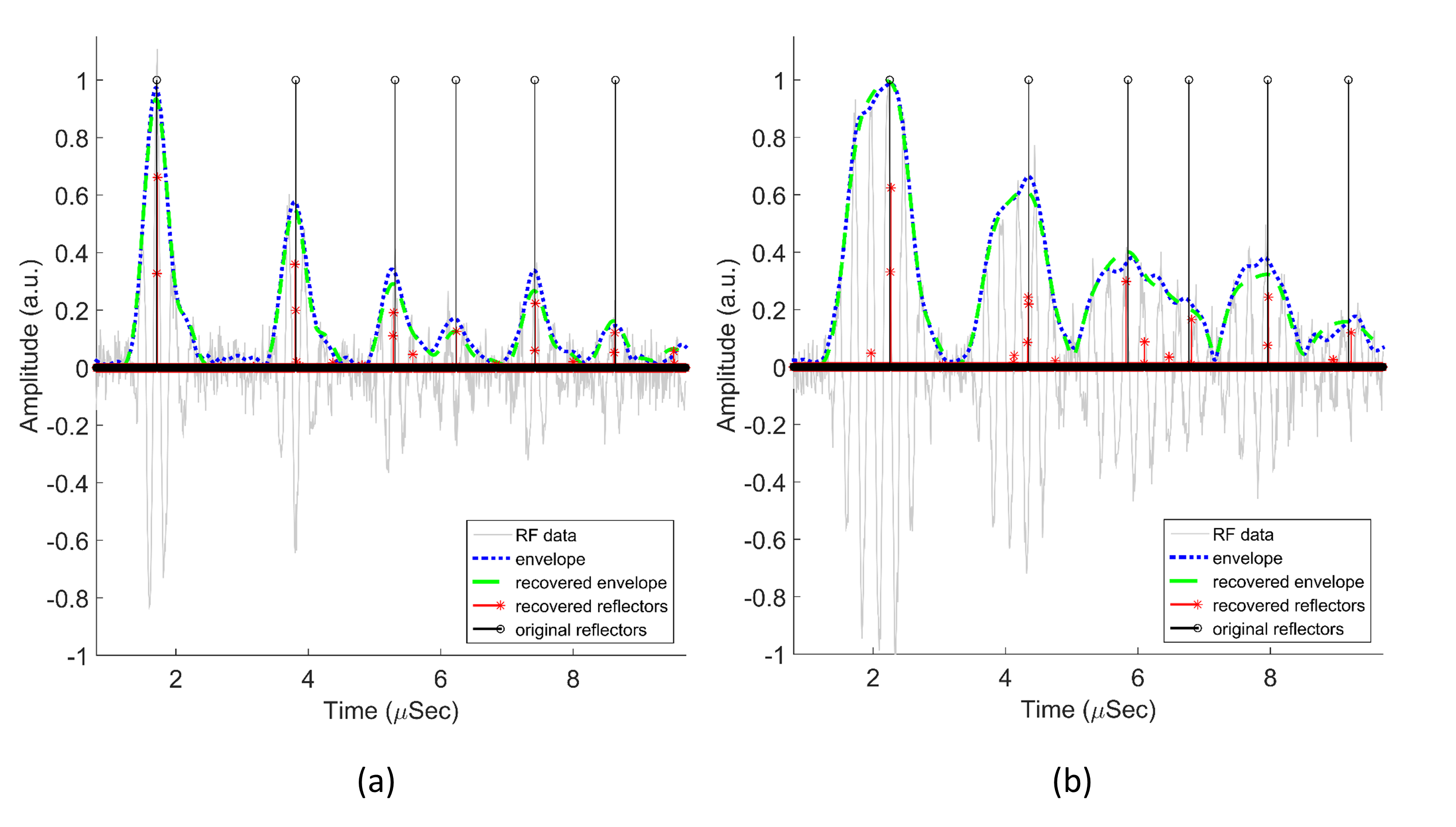}\caption{\label{fig:Results-of-phantom}Representative results of phantom experiments:
reflectivity estimates are presented against true reflectors' location,
for two different pulse lengths at SNR = 23dB. (a) Results from a
single cycle pulse (b) Results from a four cycle pulse. Each subplot
contains the received RF-line, the estimated and the original reflectivity
functions and the original and reconstructed envelopes of the signal. }
\end{figure*}

In order to quantify the performances of the localization algorithm,
the localization error was measured for different levels of SNR and
for pulses containing different number of cycles. 100 experiments
were performed for each SNR level and pulse length combination. A
representative example of an estimation performed with different SNR
levels can be found in Figure \ref{fig:Localization-loss}. As the
SNR decreases the localization error increases. For sufficiently large
SNR, the support of the solution (of the discrete optimization problem)
is tightly clustered around the true support (on the continuum). This
phenomenon is typical for $\ell_{1}$ minimization problems (see the
analysis on compact domains in \cite{duval2013exact,duval2015sparse}).
Additionally, we present a comparison between our convex optimization
approach and two widespread algorithms: MUSIC and orthogonal matching
pursuit (OMP). The convex optimization approach shows competitive
results and outperforms these algorithms for low SNR values. We stress
that we do not aim to make a thorough performance comparison of different
algorithms for the decomposition of stream of pulses. 

It is important to note that the measured SNR reflects the SNR level
for the closest reflector. The signal resulting from reflectors located
further from the transducer has lower SNR. As the noise level increases
the signal reflected from some of the deeper reflector drops below
the noise floor. Our main result in Theorem \ref{th:main} do not
guarantee the detection and localization of the reflectors in these
cases. 

The reduced localization for low SNR levels and long pulses is reflected
by higher standard deviation of the localization error (see Figure
\ref{fig:Localization-loss}(c)). When transmitting short pulses the
localization of the reflectors remains stable even for low SNR levels.
The standard deviation of the localization error for each SNR level
is higher for longer pulses, as predicted. 

\begin{figure*}
\begin{centering}
\includegraphics[scale=0.2]{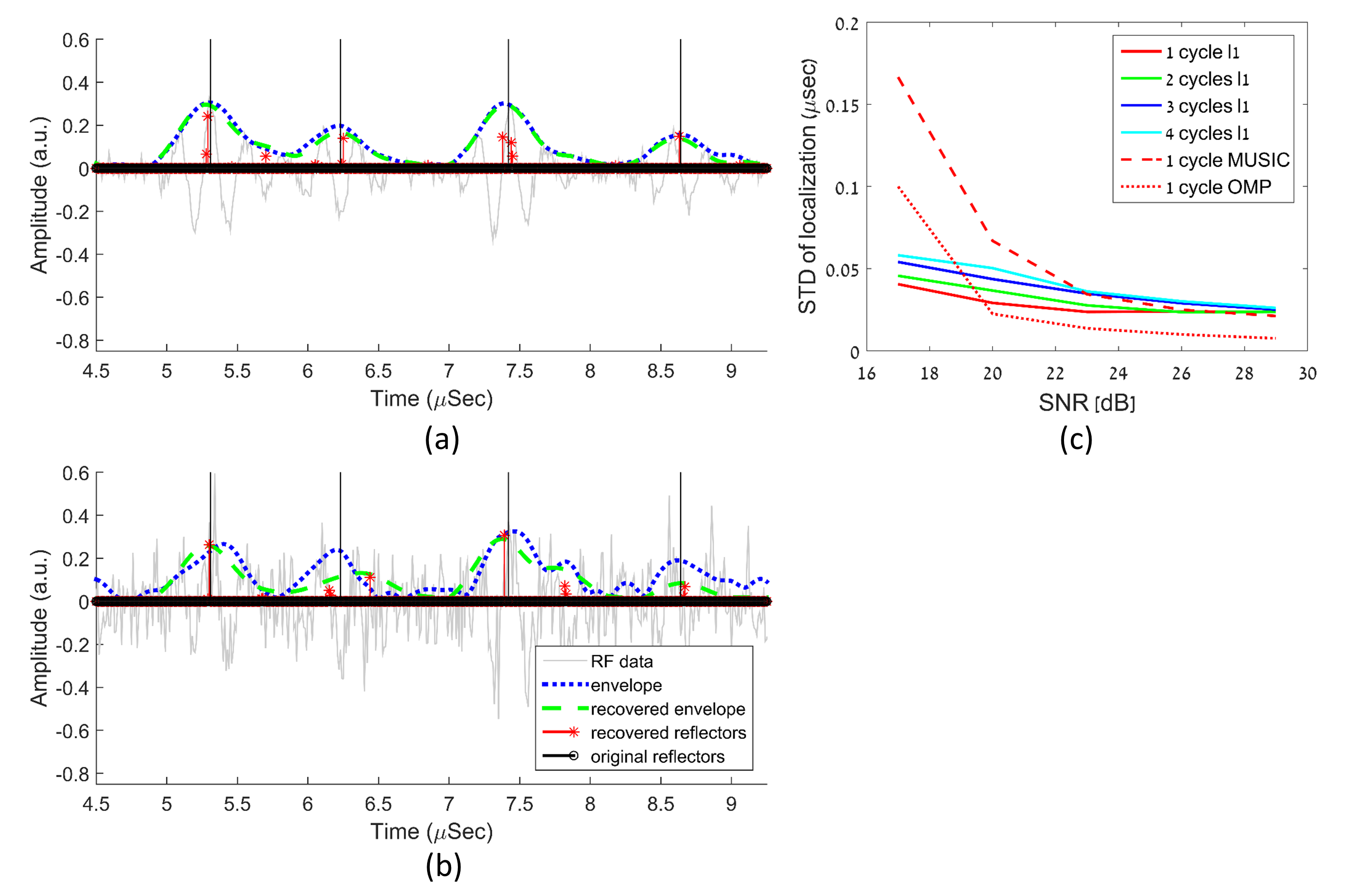}
\par\end{centering}

\caption{\label{fig:Localization-loss}Localization reduction for low SNR levels:
reflectivity estimates of the 3rd - 6th reflectors are presented against
true reflectors location, for single cycle pulse at two different
SNR values. Subplot (a) shows the results for SNR = 29dB and subplot
(b) shows the results for SNR = 17dB. Subplot (c) presents the standard
deviation of the localization error of (\ref{eq:opt_noisy}) for different
pulses as function of the SNR. Additionally, the errors of OMP and
MUSIC algorithms are presented for a single cycle pulse. 100 experiments
were performed for each case.}
\end{figure*}

Following Corollary \ref{cor:artifacts}, the solution $\hat{x}$
can include some false reflectors detections of small amplitude, far
from the real location of the reflectors. In order to quantify the
amplitude of these estimation artifacts as a function of the SNR level,
the mean value of $\left\Vert \hat{x}\right\Vert _{1}$ between the
first and the second reflector was evaluated (see gray region in Figure
\ref{fig:Estimation-artifacts}(a)). The minimal distance from each
reflector used in this analysis was two wavelength. Since the distance
between the first and second reflections in the measured (low-resolution)
signal is reduced for longer pulses, the analysis was performed for
the single cycle measurements. A monotonic reduction in the mean amplitude
of the false detections as a function SNR was measured as presented
in Figure \ref{fig:Estimation-artifacts}(b). This inverse relation
between the SNR and the estimation artifacts corroborates Corollary
\ref{cor:artifacts}.

\begin{figure*}
\begin{centering}
\includegraphics[scale=0.7]{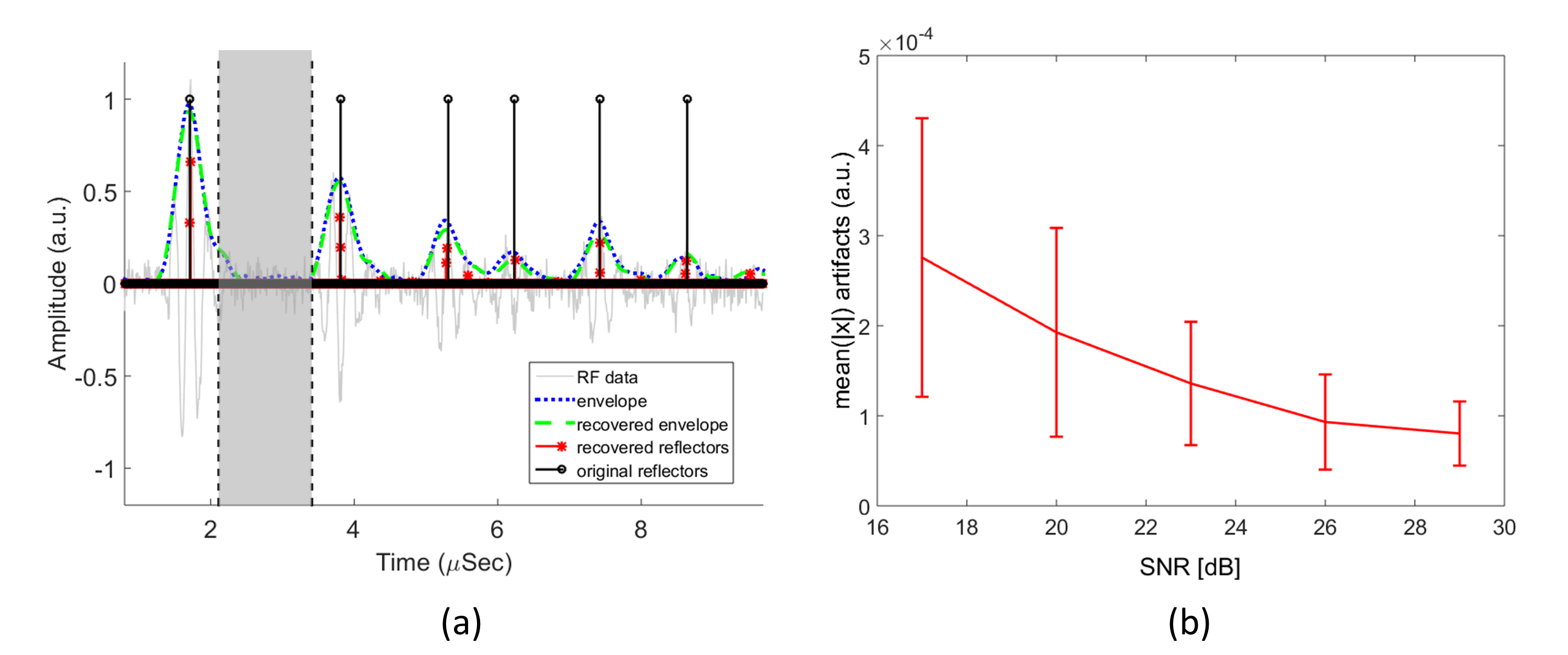}
\par\end{centering}

\caption{\label{fig:Estimation-artifacts}False detections: The mean value
of non-zero reflectivity estimates located between the first and second
reflector but far from the true location of these reflectors (gray
region in subplot (a)) was evaluated in order to quantify the estimation
artifacts level (subplot (b)) as a function of the SNR. }
\end{figure*}

\section{Conclusions \label{sec:Conclusions}}

In this work, we derived theoretical results regarding the localization
properties of convex optimization technique for decomposition of a
stream of pulses. This highlights the appealing properties of the
convex optimization approach for the decomposition of stream of pulses,
as proved and discussed in previous works \cite{bendorySOP,bendoryPositiveSOP}. 

From the theoretical perspective, we demonstrated the effectiveness
of the separation condition. Furthermore, it is clear that without
some separation, no algorithm can decompose a stream of pulses in
the presence of noise. However, it is not clear yet that a separation
condition as presented in Definition \ref{def:separation} is necessary,
and what is the optimal regularity condition for the recovery by tractable
algorithms, such as convex optimization, algebraic methods and iterative
algorithms. 

The experimental results presented in Section \ref{sec:exp} corroborate
our theoretical results. With sufficient separation of scatterers,
the localization degree of strong scatterers was shown to be governed
by the noise: the mean and the standard deviation of the localization
error increases with the noise level. In both numerical simulations
and phantom experiments, for a given SNR level, short pulses produced
better localization compared to wider pulses. In addition, the amplitude
of artifact scatterers located far from the location of the real scatterers
is shown to be small and decrease monotonically with the noise level.

The results presented in this paper have important implications on
studies attempting to perform super-resolution on signals that can
be modeled as streams of pulses. Specifically, the results of ultrasound
deconvolution algorithms based on signal sparsity/compressibility
priors can be analyzed and explained. The reliable localization of
strong reflectors in ultrasound images has important applications
such as the accurate measurement of the intima-media thickness (IMT),
the innermost layers of the of an artery, in order to facilitate early
detection of plaque formation \cite{wendelhag1997new,wikstrand2007methodological}.
Another possible use could be in non-destructive tests. Moreover,
the stream of pulses model is a popular model in compressed sensing
ultrasound \cite{wagner2012compressed,DBLP:journals/corr/ChernyakovaE13}
and compressed sensing radar \cite{Bar-IlanSub-Nyquis}. Therefore,
the presented results have implication on the interpretation of the
representations of the signals in these applications.

\textbf{Acknowledgment} The authors thank Amit Livneh for his help
with the in-vitro ultrasound setup and to the anonymous referees for
their constructive and valuable comments that have significantly improved
this paper. 

\bibliographystyle{plain}
\addcontentsline{toc}{section}{\refname}\bibliography{bib}

\appendix

\section{\label{sec:appA}Proof of Lemma \ref{lemma:2.1}}

Let $u_{m}=sign\left(c_{m}\right)$, where $sign(z):=\frac{z}{\left|z\right|}$
and $q(t)$ the corresponding interpolating function as defined in
Lemma \ref{lemma:q}. Denoting $q[k]=q\left(k/N\right),\thinspace k\in\mathbb{Z}$
we get from Lemma \ref{lemma:q} for all $k\in\mathbb{Z}_{near}$
\[
\left|q\left[k\right]\right|\leq1-\frac{\beta\left(k-k_{m}\right)^{2}}{4g(0)\left(N\sigma\right)^{2}}.
\]
Therefore, we conclude that 
\begin{eqnarray}
 & \sum_{k\in\mathbb{Z}_{near}}q[k]\hat{x}[k]\nonumber \\
= & \sum_{k\in\hat{K}_{near}}q[k]\hat{x}[k]\nonumber \\
\leq & \sum_{\left\{ n:\thinspace\hat{k}_{n}\in\hat{K}_{near}\right\} }\left|\hat{c}_{n}\right|\left|q\left[\hat{k}_{n}\right]\right|\nonumber \\
\leq & \sum_{\left\{ n:\thinspace\hat{k}_{n}\in\hat{K}_{near}\right\} }\left|\hat{c}_{n}\right|\left(1-\frac{\beta d^{2}\left(\hat{k}_{n},K\right)}{4g(0)\left(N\sigma\right)^{2}}\right).\label{eq:low_bound}
\end{eqnarray}
Using the fact that $\left\Vert q\right\Vert _{\infty}:=\max_{k\in\mathbb{Z}}\left|q[k]\right|\leq1$
and by Theorem \ref{th:noisy} we observe that 
\begin{eqnarray}
\left|\sum_{k\in\mathbb{Z}_{near}}q[k]\left(\hat{x}[k]-x[k]\right)\right| & \leq & \sum_{k\in\mathbb{Z}}\left|q[k]\right|\left|\hat{x}[k]-x[k]\right|\nonumber \\
 & \leq & \left\Vert q\right\Vert _{\infty}\left\Vert \hat{x}-x\right\Vert _{1}\label{eq:qh}\\
 & \leq & \frac{16\gamma^{2}}{\beta}\delta.\nonumber 
\end{eqnarray}
Then, combining (\ref{eq:low_bound}) and (\ref{eq:qh}) into

\begin{eqnarray*}
\sum_{k\in\mathbb{Z}_{near}}q[k]\hat{x}[k] & = & \sum_{k\in\mathbb{Z}_{near}}q[k]\left(\hat{x}[k]-x[k]\right)+\sum_{k\in\mathbb{Z}_{near}}q[k]x[k].
\end{eqnarray*}
we get 
\begin{eqnarray*}
\sum_{\left\{ m:\thinspace k_{m}\in K\right\} }\left|c_{m}\right|-\frac{16\gamma^{2}}{\beta}\delta & \leq & \sum_{k\in\mathbb{Z}_{near}}q[k]\left(\hat{x}[k]-x[k]\right)+\sum_{k\in\mathbb{Z}_{near}}q[k]x[k]\\
 & \leq & \sum_{\left\{ n:\thinspace\hat{k}_{n}\in\hat{K}_{near}\right\} }\left|\hat{c}_{n}\right|\left(1-\frac{\beta d^{2}\left(\hat{k}_{n},K\right)}{4g(0)\left(N\sigma\right)^{2}}\right),
\end{eqnarray*}
where we used the fact that $\sum_{k\in\mathbb{Z}_{near}}q[k]x[k]=\sum_{\left\{ m:\thinspace k_{m}\in K\right\} }\left|c_{m}\right|$.
To conclude the proof, we use (\ref{eq:opt_noisy}) to derive 
\[
\left\Vert x\right\Vert _{1}=\sum_{\left\{ m:\thinspace k_{m}\in K\right\} }\left|c_{m}\right|\geq\left\Vert \hat{x}\right\Vert _{1}=\sum_{\left\{ n:\thinspace\hat{k}_{n}\in\hat{K}\right\} }\left|\hat{c}_{n}\right|\geq\sum_{\left\{ n:\thinspace\hat{k}_{n}\in\hat{K}_{near}\right\} }\left|\hat{c}_{n}\right|.
\]
Therefore 
\[
\sum_{\left\{ n:\thinspace\hat{k}_{n}\in\hat{K}_{near}\right\} }\left|c_{n}\right|-\frac{16\gamma^{2}}{\beta}\delta\leq\sum_{\left\{ n:\thinspace\hat{k}_{n}\in\hat{K}_{near}\right\} }\left|\hat{c}_{n}\right|\left(1-\frac{\beta d^{2}\left(\hat{k}_{n},K\right)}{4g(0)\left(N\sigma\right)^{2}}\right)
\]
and 
\begin{eqnarray*}
\sum_{\left\{ n:\thinspace\hat{k}_{n}\in\hat{K}_{near}\right\} }\left|\hat{c}_{n}\right|d^{2}\left(\hat{k}_{n},K\right) & \leq & \frac{64g(0)\gamma^{2}\left(N\sigma\right)^{2}}{\beta^{2}}\delta\\
 & \leq & \frac{64g(0)\gamma^{4}}{\beta^{2}}\delta.
\end{eqnarray*}
This completes the proof.
\end{document}